\documentclass[10pt,a4paper]{article}
\usepackage[english]{babel} 
\usepackage{amsmath}
\usepackage{mathtools}
\usepackage{amsfonts}
\usepackage{amssymb}
\usepackage{amsthm}
\usepackage{makeidx}
\usepackage{graphicx}
\usepackage{natbib}
\usepackage{setspace}
\usepackage{float}
\usepackage{hyperref}
\usepackage[margin=1in]{geometry}
\usepackage{setspace} 
\usepackage[titletoc]{appendix}

\usepackage{lineno}[pagewise]
\usepackage[noabbrev]{cleveref}
\usepackage[inline]{enumitem}

\hypersetup{
	colorlinks   = true,
	citecolor    = blue,
	urlcolor = blue
}

\author{Marcelo Brutti Righi$^{a,}$\footnote{Corresponding author. $^{a}$Business School, Federal University of Rio Grande do Sul, Washington Luiz, 855, Porto Alegre, Brazil, zip 90010-460.}\\  \small{ \href{mailto:marcelo.righi@ufrgs.br}{marcelo.righi@ufrgs.br}} 
	\and Marlon Ruoso Moresco$^{a}$\\\small{\href{mailto:marlon.moresco@ufrgs.br}{marlon.moresco@ufrgs.br}}}

\date{}
\title{Star-Shaped deviations}

\newtheorem{Def}{Definition}
\newtheorem{Thm}[Def]{Theorem}
\newtheorem{Prp}[Def]{Proposition}

\newtheorem{Crl}[Def]{Corollary}
\theoremstyle{definition}
\newtheorem{Exm}[Def]{Example}

\theoremstyle{remark}
\newtheorem{Rmk}[Def]{Remark}

\DeclareMathOperator*{\essinf}{ess\,inf}
\DeclareMathOperator*{\esssup}{ess\,sup}

\begin{document}
	
	\maketitle
	\begin{abstract}
		We propose the Star-Shaped deviation measures in the same vein
		as Star-Shaped risk measures and Star-Shaped acceptability indexes. We characterize Star-Shaped deviation measures through Star-Shaped 
		acceptance sets and as the minimum of a family of Convex deviation measures.  We also expose an interplay between Star-Shaped risk measures and deviation measures.
	\end{abstract}	
	
	\section{Introduction}

	Since \cite{Artzner1999} the axiomatization of risk measures has gained space in the literature. Their seminal paper argues that a ``coherent'' risk measure should satisfy four properties, among them, Positive Homogeneity. This property implies that the risk of a position is proportional to its size, i.e., for a risk measure $\rho$ and a positive real $\lambda$, it follows that $\rho(\lambda X) = \lambda \rho (X).$ However, Positive Homogeneity quickly came under criticism, mainly because the size of a financial position can affect the position risk due to liquidity risk, i.e., potential losses from difficulty into negotiating larger positions. In this sense, \cite{Follmer2002}, \cite{Frittelli2002}, \cite{Acerbi2008} and \cite{Lacker2018} argue against the
	Positive Homogeneity and sub-additivity assumptions adopted in the framework of coherent risk measures by focusing on Convexity, since, under $\rho(0)\leq0$, it implies that for $\lambda \geq 1,$ $\rho(\lambda X) \geq \lambda \rho (X)$. 
	
	Nonetheless, Convexity is actually stronger than purely demanding  $\rho(\lambda X) \geq \lambda \rho (X)$ for $\lambda \geq 1$. The latter is called Star-Shapedness and is the focus of this study. 
	In this sense, \cite{Castagnoli2021} proposes the class of Star-Shaped risk measures. This nomenclature comes from the Star-Shaped property of the generated acceptance set. The reasoning for Star-Shapedness as sensible axiomatic requirement is that if a position is acceptable, any scaled reduction of it also is. The key-point in this theory is that a monetary risk measure is Star-Shaped if and only if it is the minimum of a family of Convex risk measures. This class gained some attention in the literature when  
	\cite{Liebrich2021} explores allocations of Star-Shaped risk measures,  \cite{Moresco2021} relate them to the broader class of monetary risk measures, \cite{Herdegen2021} consider portfolio optimization and arbitrage, and \cite{Righi2021b} explores the interplay with Star-Shaped acceptability indexes. 
	
	The arguments exposed in their work affect deviation measures in the same way as they affect monetary risk measures. Such concept of deviation is axiomatized for Convex functionals in \cite{Rockafellar2006}, \cite{Pflug2006} and \cite{Grechuk2009}. The main idea is to consider generalizations of the standard deviation and similar measures in an axiomatic fashion. See \cite{Pflug2007} and \cite{Rockafellar2013} for a comprehensive review. Recently, \cite{Righi2016}, \cite{Righi2018a} and \cite{Righi2018b} explore advantages of a more complete analysis that considers both risk and deviation measures. 
	
	Thus, it is reasonable to generalize both Positive Homogeneity and Convexity for deviation measures by Star-Shapedness. Therefore, in this paper, we define and explore the class of Star-Shaped deviation measures. In doing so, we obtain that a deviation measure is Star-Shaped if and only if it is the minimum of a family of Convex deviation measures. This result is obtained with distinct techniques from the one for risk measures since deviation measures do not fulfill the property of Monotonicity, which is crucial in the paper of \cite{Castagnoli2021}. We develop a concept of acceptance set for deviation measures and show an interplay between Star-Shapedness for such sets and deviation measures. We also expose an interplay between Star-Shaped risk measures and deviation measures.
	
  The intuitive reasoning is similar to the one for monetary risk measures since a reduction by scaling of a position must not produce larger deviations. Further, \cite{Righi2018a} and \cite{Nendel21} study, respectively, compositions and decompositions of required capital and insurance premium into a risk measure and a deviation
	measure, which can be understood as the actuarial safety margin. Then, the representation of a Star-Shaped deviation as the minimum of Convex deviations
	suggests that using Star-Shaped deviation measures decreases the premium costs.

	We consider a probability space $(\Omega,\mathcal{F},\mathbb{P})$. All equalities and inequalities are in the $\mathbb{P}-a.s.$ sense.  
	% 	Let $L^0=L^0(\Omega,\mathcal{F},\mathbb{P})$ and 
	Let $L^{\infty}=L^{\infty}(\Omega,\mathcal{F},\mathbb{P})$ be the spaces of (equivalence classes under $\mathbb{P}-a.s.$ equality of) finite and essentially bounded random variables. We consider in  $L^\infty$ its strong topology. We say that a set $\mathcal{A} \subseteq L^\infty $ is Star-Shaped if $ \lambda \mathcal{A} \subseteq \mathcal{A}$ for all $\lambda \in [0,1]$, or equivalently, $X \in \mathcal{A}$ implies that $\lambda X \in \mathcal{A}$ for all $\lambda \in [0,1]$.  We denote, by $E[X]$, $F_{X}$, and $F_{X }^{-1}$, the expected value, the (increasing and right-continuous) cumulative probability function, and its left quantile for $X\in L^\infty$ with respect to $\mathbb{P}$.  We denote that $X$ and $ Y$ have the same distribution by $X \sim Y$. The notation $X\succeq Y$, for $X,Y\in L^\infty$, indicates stochastic dominance. For second-order it means $E[f(X)]\leq E[f(Y)]$ for any increasing Convex function $f\colon\mathbb{R}\rightarrow\mathbb{R}$, while for Convex order it means $E[X]=E[Y]$ and $E[f(X)]\leq E[f(Y)]$ for any Convex function.
We define $1_A$ as the indicator function for an event $A\in\mathcal{F}$. We identify constant random variables with real numbers.
	% 	  A pair $X,Y\in L^0$ is called coMonotone if $\left( X(w)-X(w^{\prime})\right)\left( Y(w)-Y(w^{\prime}) \right)\geq0,\:\:w,w^{'}\in\Omega$ holds $\mathbb{P}\otimes\mathbb{P}-a.s.$ We denote by $X_n\rightarrow X$ convergence in the $L^\infty$ essential supremum norm $\lVert \cdot\rVert_{\infty}$, whereas $\lim\limits_{n\rightarrow\infty}X_n=X$ indicates $\mathbb{P}-a.s.$ convergence. The notation $X\succeq Y$, for $X,Y\in L^\infty$, indicates second-order stochastic dominance, that is, $E[f(X)]\leq E[f(Y)]$ for any increasing Convex function $f\colon\mathbb{R}\rightarrow\mathbb{R}$. 
	%Let $\mathcal{P}$ be the set of all probability measures on $(\Omega,\mathcal{F})$. We denote, by $E_{\mathbb{Q}}[X]=\int_{\Omega}Xd\mathbb{Q}$, $F_{X, \mathbb{Q}}(x)=\mathbb{Q}(X\leq x)$, and $F_{X, \mathbb{Q}}^{-1}(\alpha)=\inf\left\lbrace x:F_{X, \mathbb{Q}}(x)\geq\alpha\right\rbrace $, the expected value, the (increasing and right-continuous) probability function, and its left quantile for $X\in L^\infty$ with respect to $\mathbb{Q}\in\mathcal{P}$. We write $X\overset{\mathbb{Q}}\sim Y$ when $F_{X,\mathbb{Q}}=F_{Y,\mathbb{Q}}$. We drop subscripts indicating probability measures when $\mathbb{Q}=\mathbb{P}$. Furthermore, let $\mathcal{Q}\subset\mathcal{P}$ be the set of probability measures $\mathbb{Q}$ that are absolutely continuous with respect to $\mathbb{P}$, with Radon--Nikodym derivative $\frac{d\mathbb{Q}}{d\mathbb{P}}$. 
	
	%{\color{red}Theorem 5 of \cite{Castagnoli2021}, Theorems 11 and 12 in \cite{Castagnoli2021}]\label{Thm:LI}}
	
	We begin by exposing the theoretical properties that appear in the literature regarding deviation measures, and we consider them in this paper.
	
	\begin{Def}\label{def:dev}
		A functional $\mathcal{D}:L^\infty\rightarrow\mathbb{R}_+ \cup\{ \infty\}$ is a deviation measure. It may fulfill the following properties:
		\begin{enumerate}
			\item Non-Negativity: For all $X\in L^\infty$, $\mathcal{D}(X)=0$ for constant $X$ and $\mathcal{D}(X)>0$ for
			non-constant X;
			\item Translation Insensitivity: $\mathcal{D}(X+C)=\mathcal{D}(X),\:\forall \:X\in L^\infty,\:\forall\:C \in\mathbb{R}$;
			\item Convexity: $\mathcal{D}(\lambda X+(1-\lambda)Y)\leq \lambda \mathcal{D}(X)+(1-\lambda)\mathcal{D}(Y),\:\forall\: X,Y\in L^\infty,\:\forall\:\lambda\in[0,1]$;
			\item Positive Homogeneity: $\mathcal{D}(\lambda X)=\lambda \mathcal{D}(X),\:\forall\:X\in L^\infty,\:\forall\:\lambda \geq 0$;
			\item Star-Shapedness:  $\mathcal{D}(\lambda X)\geq \lambda\mathcal{D}(X),\:\forall\: X\in L^\infty,\:\forall\:\lambda \geq 1$.
			\item Lower Range Dominance:   $\mathcal{D}(X)\leq E[X]-\operatorname{ess}\inf X,\:\forall\:X\in L^\infty$;
			\item Law Invariance: If $F_X=F_Y$, then $\mathcal{D}(X)=\mathcal{D}(Y),\:\forall\:X,Y\in L^\infty$.
		\end{enumerate}
		A deviation measure $\mathcal{D}$ is called proper if it fulfills (i) and (ii);  Convex if it is proper and respects (iii); generalized (also called Coherent) if it is Convex and fulfills (iv); Star-Shaped if it is proper and fulfills (v); Lower Range Dominated if it satisfies (vi) and Law Invariant if it has (vii).
	\end{Def}
	
	%Non-Negativity assures there is dispersion only for non-constant positions. Translation Insensitivity indicates that the deviation does not change if a constant value is added. Note that deviations assume counter domain on $\mathbb{R}_+$, where negative values are not present. In this sense, the case of $\mathcal{D}(X)=0$ represents the absence of variability and reflects the best possible situation from the point of view for deviations. Such axioms contained in this definition are related to the concept of norm, as in \cite{Righi2017}, which explores a symmetrization for deviation measures. 
	
	\begin{Rmk}
		It is straightforward to prove that, for any $X\in L^\infty$, the following is equivalent:
		\begin{enumerate*}[label=(\roman*)]
			\item $\mathcal{D}$ is Star-Shaped;
			\item  $\mathcal{D}(\lambda X)\leq\lambda\mathcal{D}(X)$ for any $0\leq\lambda\leq 1$;
			\item $\lambda\to\frac{\mathcal{D}(\lambda X)}{\lambda}$ is non-decreasing.
		\end{enumerate*}
		Moreover, for a proper sub-additive deviation measure, i.e. $\mathcal{D}(X+Y)\leq\mathcal{D}(X)+\mathcal{D}(Y),\:\forall\:X,Y\in L^\infty$, it is easy to see the equivalence between the following:
		\begin{enumerate*}[label=(\roman*)]
			\item $\mathcal{D}$ is Star-Shaped;
			\item  $\mathcal{D}$ is Positive Homogeneous;
			\item $\mathcal{D}$ is a generalized deviation measure.
		\end{enumerate*}
	\end{Rmk}

	\begin{Exm}\label{Exm:devmeas}
		We now expose some examples of Star-Shaped deviation measures. Recall that for proper deviation measures, Star-Shapedness is implied by both Positive Homogeneity and Convexity. Further, some measures below do not satisfy $D(X)>0$ for non-constant $X$, but only $D(C)=0$ for any $C\in\mathbb{R}$. In some papers, such as \cite{Bellini2020}, they are called
		variability measures or dispersion measures. Such maps can attend Non-Negativity when added to a proper deviation measures.
		\begin{enumerate}
			\item Standard Deviation (SD): This is perhaps the most well-known measure of variability, being defined as $SD(X)=E\left[(X-E[X])^2\right]^{\frac{1}{2}}$. It is a generalized deviation measure that represents the second moment around expectation and has been considered a proxy for risk in modern finance since the pioneering work of \cite{Markowitz1952}. The SD inspires the whole conception of deviation measures, once the symmetry is dropped. This is important as dispersion from gains and losses have distinct impacts. The asymmetric forms of the SD are the lower and upper semi-deviations (SD$_{-}$/SD$_{+}$). They consider dispersion only from values, respectively, below or above the expectation to avoid symmetry. This is necessary as not all dispersion in a financial position is undesirable, in fact, a result above its expected return is in general beneficial. They are defined as $SD_-(X)=(E\left[((X-E[X])^-)^2\right])^{\frac{1}{2}}$ and $SD_+(X)=(E\left[((X-E[X])^+)^2\right])^{\frac{1}{2}}$.
			\item Full Range (FR): This extremely conservative generalized deviation measure is defined as $FR(X)=\esssup X-\essinf X$ and represents the larger possible difference for two realizations of $X$. 
			Due to the conservatism of the FR, Lower and Upper Range (LR/UR) arise as adaptations to consider the range below or above the expectation, respectively. They are defined as $LR(X)=E[X]-\essinf X$ and $UR(X)=\esssup X-E[X]$. The idea is similar to that for SD$_{-}$ and SD$_{+}$.

			\item Loss Value at Risk Deviation (LVaRD): This is a concept derived from the  Loss VaR of \cite{Bignozzi2020}, it is defined as $LVaRD_\alpha (X) = \sup_{u \geq 0}\{-F^{-1}_{(X-E[X])}(\alpha(u)) -u\} , $ where $ \alpha \colon [0,\infty) \rightarrow (0,1]$  is an increasing and right-continuous function which represents some benchmark loss. It is easy to check that, LVaRD is a variability measure not Convex and neither it is positively homogeneous unless $\alpha$ is constant. However,  it is Star-Shaped and Translation Insensitive.

			\item Regular based Deviation (RbD): Let $f:L^\infty\rightarrow\mathbb{R}$ be Monotone and Star-Shaped, with   $f(X)\geq -E[X]$ and $f(X) = -E[X]$ if and only if $X$ is constant. Then we have that $\mathcal{D}_f(X)=f(X-E[X])$ is the Star-Shaped deviation induced by $f$. See  \Cref{prp:dev_risk} for a concrete example. Under the same $f$, we can define the  Lower and Upper Regular based Deviation (LD/UD) as $LD_f(X)=f((X-E[X])^-)$ and   $UD_f(X)=f((X-E[X])^+)$, the intuition behind those is the same as for the semi-deviations and Upper/Lower Ranges. Note that all the previous examples of deviations are special cases of this approach. In the first example $f$ is the $L^2$ norm, where $\mathcal{D}_f$ is the standard deviation, $SD_- = LD_f$ and $SD_+ = UP_f$. In the second example for the LR, we have $f(X)=-\essinf X$ and for UR, $f (X) = \esssup X$. Lastly, in the LVaRD, we have the risk measure LVaR doing the role of $f$. Further,  we can define $\mathcal{D}(X) = \min (UD_f (X), LD_f (X))$. This is a regularization of the ranges and is a Star-Shaped deviation measure.
			
			\item  Loss-Deviation (LD): This measure is linked to the dispersion of results worse than a benchmark, typically a risk measure, measured by usual $p$-norms. This concept is explored by \cite{Righi2016} and \cite{Righi2019}. Let $f : L^{\infty}\rightarrow \mathbb{R}$ be Monotone and Positive Homogeneous such that  $f(X+c)=f(X)+c$ for any $X\in L^\infty$ and any $c\in\mathbb{R}$. Then, and  its loss-deviation is $LD(X)= \lVert (X-f(X))^-\rVert_p,\:p\in[1,\infty]$. This deviation is a generalization of the lower semi-deviation, and it is not Convex for any concave $f$, except for the negative expectation. However, it is Star-Shaped since it is Positive Homogeneous.
			\item Minkowski Deviation ($MD_\mathcal{A}$): Given an acceptance set $\mathcal{A}$ that is Star-Shaped, radially bounded for non-constants and stable under scalar addition (see \cite{Moresco2021b} for precise definition and financial intuition of those properties), the Minkowski deviation is defined as $MD_\mathcal{A} (X) = \inf \{ m >0  \colon \frac{X}{m} \in \mathcal{A} \}$. Any Positive Homogeneous deviation measure $\mathcal{D}$ is a Minkowski Deviation by taking $\mathcal{A} = \{X \in L^\infty \colon \mathcal{D} (X) \leq 1\} $. MD is Convex if and only if the acceptance set also is Convex. Related to our main result in \Cref{prp:dev}, by letting $\mathcal{A}_Y = \{\lambda Y + c \colon \lambda \in [0,1], c \in \mathbb{R} \}$,  Proposition 3.1 of \cite{Moresco2021b} gives $MD_\mathcal{A} (X) = MD_{\bigcup_{Y \in \mathcal{A}} \mathcal{A}_Y} (X) = \inf_{Y \in \mathcal{A}} MD_{\mathcal{A}_Y} (X).$ Nonetheless, their framework can not embrace deviation measures that are not Positive Homogeneous, in particular those that are the focus of this study.
			
			\item Iterquantile Deviation (IQD): Based on Value at Risk, defined as $VaR^\alpha (X) = -F_X^{-1}(\alpha), \alpha \in (0,1)$, the IQD  is a commonly used measure of dispersion in statistics, it measures the distance between two quantiles, for $\alpha \in (0,0.5) $ we have that $IQD^\alpha (X) = VaR^\alpha (X)- VaR^{1-\alpha} (X)$. This measure is not Convex, while it is Positive Homogeneous.  Hence, it is Star-Shaped. Furthermore, while it is Translation Insensitive, it is only Non-Negative under a mixture with another proper deviation measure. This measure is studied in \cite{Wang2020c} and \cite{Bellini2020}. We have that $D^\alpha(X) = (IQD^\alpha (X) )^2 +SD(X) $ is a non-trivial, Star-Shaped deviation measure that is neither Convex nor Positive Homogeneous.
			
			\item Inter-ES Deviation (IED): This deviation is similar to the IQD, however, here the Expected Shortfall (ES) does the same role as the VaR in the IQD. For $\alpha \in (0,1)$, the ES is defined as $ES^\alpha (X) = \frac{1}{\alpha} \int_0^\alpha VaR^s (X) ds $. Then, we have that $IED^\alpha (X) = ES^\alpha (X)-ES^{1-\alpha} (X) $. This deviation measure in convex-order consistent, Law Invariant and has all properties of a generalized deviation measure, with exception for Non-Negativity. Again, this can be easily solved by adding it to a proper deviation measure. Furthermore, $IED^\alpha$ is the smallest Law Invariant, Translation Insensitive Convex functional dominating $IQD^\alpha$, see \cite{Wang2020c} Theorem 5 and Example 7. Again, we can easily derive a Star-Shaped deviation measure that is neither Convex nor Positive Homogeneous by squaring the IED and adding it to a Star-Shaped deviation measure.
			
		\end{enumerate}
	\end{Exm}

	\section{Results}\label{sec:results}
	
	\Cref{accep set} below is a direction on how to extend the concept of acceptance set for monetary risk measures to the framework of deviation measures. A more extensive study on such acceptance sets is beyond the scope of this paper and will be postponed for future research. Regarding an alternative approach based on Positive Homogeneity, see \cite{Moresco2021b}.
	
	\begin{Prp}\label{accep set}
		The following is equivalent for a deviation measure $\mathcal{D}\colon L^\infty\to\mathbb{R}_+\cup\{\infty\}$:
		\begin{enumerate}
			\item $\mathcal{D}$ is Star-Shaped.
			\item $\mathcal{A}_\mathcal{D} \coloneqq \{X\in L^\infty\colon \mathcal{D}(X)\leq E[X]\}$ is Star-Shaped.
			\item There is a Star-Shaped set $\mathcal{A}$ such that $\mathcal{D}(X)=  \mathcal{D}_{\mathcal{A}}(X) \coloneqq \inf\{m\in\mathbb{R}\colon X+m\in\mathcal{A}\}+E[X]$, for all $X \in L^\infty$.
		\end{enumerate}
	In this case we have that $\mathcal{D}_{\mathcal{A}_D}=\mathcal{D}$ and $\mathcal{A}\subseteq\mathcal{A}_{\mathcal{D}_\mathcal{A}}$.
	\end{Prp}
	
	\begin{proof}
		(i)$\implies$(ii). Note that $X\in\mathcal{A}_\mathcal{D}$ if and only if $0 \leq \mathcal{D}(X) \leq E[X]$. Then, for any $X \in \mathcal{A}$ and  $\lambda\in[0,1]$ we have that $\mathcal{D}(\lambda X)\leq \lambda \mathcal{D}(X) \leq \lambda E[X]$. Thus, $\lambda X \in \mathcal{A}$ which implies that $\mathcal{A}_\mathcal{D}$ is Star-Shaped.
		
		(ii)$\implies$(iii). Let $\mathcal{A}=\mathcal{A}_\mathcal{D}$. Then
		\begin{align*}  \mathcal{D}_{\mathcal{A}} (X) - E[X] &= \inf\{m\in\mathbb{R}\colon X+m\in\mathcal{A}_\mathcal{D}\}\\&=\inf\{m\in\mathbb{R}\colon-E[X+m]+\mathcal{D}(X+m)\leq 0\}\\
		&=\inf\{m\in\mathbb{R}\colon-E[X]+\mathcal{D}(X)\leq m\}\\
		&=\mathcal{D}(X) - E[X].
		\end{align*}
		
		(iii)$\implies$(i). Let $\lambda\in(0,1]$ the case where $\lambda = 0$ is trivial. Note that if $\mathcal{A} \subseteq \frac{1}{\lambda} \mathcal{A}$, then for any $X\in L^\infty$ we have that \begin{align*}
		\mathcal{D}(\lambda X)&=\inf\{m\in\mathbb{R}\colon \lambda X+m\in\mathcal{A}\}+E[\lambda X]\\
		&=\lambda(\inf\{m\in\mathbb{R}\colon  \lambda(X+m)\in\mathcal{A}\}+E[X])\\
		&\leq\lambda(\inf\{m\in\mathbb{R}\colon  X+m\in\mathcal{A}\}+E[X])=\lambda\mathcal{D}(X).
		\end{align*}
	Moreover,	$\mathcal{D}_{\mathcal{A}_D}=\mathcal{D}$ is trivial from previous items. Further, let $X\in\mathcal{A}$. Then \[0\geq\inf\{m\in\mathbb{R}\colon X+m\in\mathcal{A}\}=\mathcal{D}_\mathcal{A}(X)-E[X].\]
	Hence, $X\in\mathcal{A}_{\mathcal{D}_\mathcal{A}}$. 	

	\end{proof}

	We now show our main result below in \Cref{prp:dev} by showing a similar result to Theorem 5 in \cite{Castagnoli2021}. However, note that beyond Star-Shapedness and Convexity, which play a key role in the \nameCref{prp:dev}, we ask nothing but the most basic axioms of deviation measures, namely Non-Negativity and Translation Insensitivity. While in \cite{Castagnoli2021}, they also demand Normalization, the role of Normalization was explored in \cite{Moresco2021}.

	\begin{Thm}\label{prp:dev}
		$\mathcal{D}\colon L^\infty\to\mathbb{R}_+\cup\{\infty\}$ is a Star-Shaped deviation measure if and only if there is a family $\{\mathcal{D}_i\}_{i\in\mathcal{I}}$ of Convex deviation measures such that the representation  \begin{equation}\label{eq:dev_min}
		\mathcal{D}(X)=\min_{i\in\mathcal{I}}\mathcal{D}_i(X),\:\forall\:X\in L^\infty,
		\end{equation}
	holds.	Moreover, such family can be chosen as the one composed by the Convex deviation measures that dominate $\mathcal{D}$, i.e. \[\mathcal{I}=\{\beta\colon L^\infty\to\mathbb{R}_+\cup\{\infty\}\colon\beta\:\text{Convex deviation measure and}\:\beta\geq\mathcal{D}\}.\]
	\end{Thm}
	
	\begin{proof}
		Let $\mathcal{D}$ be given as in \eqref{eq:dev_min}, we shall show that it is a Star-Shaped deviation measure. 
		For Non-Negativity  %as $\mathcal{D}(X) = \mathcal{D}_{i^*}(X)$ for some $i^*\in\mathcal{I}$ and $\mathcal{D}_{i^*} (X) = 0 $ if and only if $X $ is constant. 
		take nonconstant $X \in L^\infty$ and $i^* \in \mathcal{I}$ such that $\mathcal{D}(X) = \mathcal{D}_{i^*} (X)$.  By Non-Negativity of $\mathcal{D}_{i^*}$, $0<\mathcal{D}_{i^*}(X) = \mathcal{D}(X)$. One argues similarly for $X \in \mathbb{R}$.
		For Translation Insensitivity we have for any $m\in\mathbb{R}$ and any $X\in L^\infty$ that 
		\[\mathcal{D}(X+m)=\min_{i\in\mathcal{I}}\mathcal{D}_i(X+m)=\min_{i\in\mathcal{I}}\mathcal{D}_i(X)=\mathcal{D}(X).\] 
		For Star-Shapedness, let $\lambda\leq 1$. Then, as each $\mathcal{D}_i$ is Convex, \[\mathcal{D}(\lambda X)=\min_{i\in\mathcal{I}}\mathcal{D}_i(\lambda X)\leq \lambda\min_{i\in\mathcal{I}}\mathcal{D}_i(X)=\lambda\mathcal{D}(X).\]
		For the converse, let $\mathcal{D}$ be Star-Shaped deviation measure, $\mathcal{A}= \mathcal{A_D}$  and $\mathcal{D_A}$ be as defined in \Cref{accep set}. We then have that $
		\mathcal{D_A}(X) =\mathcal{D}(X)$.
		
		Now we will find a family of Convex deviation measure such that \cref{eq:dev_min} holds. For any $Y\in L^\infty$  we define 
		\[\mathcal{A}_{Y}=conv(\{Y-E[Y]+\mathcal{D}(Y)\}\cup\{0\})+\mathbb{R}_+ = \{ \lambda (Y - E[Y] + \mathcal{D}(Y) ) + m \colon \lambda \in [0,1], m\geq 0\},\] 
		and let $\mathcal{D}_{Y}(X)= \mathcal{D}_{A_Y} (X)$. 
		We have that each $\mathcal{A}_{Y}$ is Convex since $\mathbb{R}_+$ is a Convex cone.  It is easy to see  that $\mathcal{D}_Y(X)=\mathcal{D}(X)=0$ for any $X$ constant. For non-constant $X$, if $X\in\mathcal{A}_Y$, then $X=\lambda(Y-E[Y]+\mathcal{D}(Y)+k)$, where $k\in\mathbb{R}_+$ and $\lambda\in[0,1]$. In this case \begin{align*}
		E[-X]+\mathcal{D}(X)&=E[-\lambda(Y-E[Y]+\mathcal{D}(Y)+k)]+\mathcal{D}(\lambda(Y-E[Y]+\mathcal{D}(Y)+k))\\
		&=-\lambda(\mathcal{D}(Y)+k)+\mathcal{D}(\lambda Y)\\
		&\leq\lambda(\mathcal{D}(Y)-\mathcal{D}(Y)-k)=-\lambda k\leq0.
		\end{align*} Thus, $X\in\mathcal{A}$, which implies  $\mathcal{A}_Y\subseteq\mathcal{A}$ and, consequently, \[\mathcal{D}_Y(X)=\inf\{m\in\mathbb{R}\colon X+m\in\mathcal{A}_Y\}+E[X]\geq\inf\{m\in\mathbb{R}\colon X+m\in\mathcal{A}\}+E[X]=\mathcal{D}(X).\]
		Thus, $\mathcal{D}(X)\leq\inf\{\mathcal{D}_Y(X)\colon Y \in L^\infty\}$. Furthermore, since $X+E[-X]+\mathcal{D}(X)\in\mathcal{A}_X$, we have that \[\mathcal{D}_X(X)=\inf\{m\in\mathbb{R}\colon X+m\in\mathcal{A}_X\}+E[X]\leq E[-X]+\mathcal{D}(X)+E[X]=\mathcal{D}(X).\] Hence, we have that $\mathcal{D}(X)=\mathcal{D}_X(X)$ and $\mathcal{D}(X)=\min\{\mathcal{D}_Y(X)\colon Y \in L^\infty \}$.
		
		Now, we need to show that each $\mathcal{D}_Y$ defines a Convex deviation measure. When $Y$ is constant $\mathcal{A}_Y = \mathbb{R}_+$ and we have that $\mathcal{D}_Y (X) = 0$ if $X$ is constant and $\mathcal{D}_Y(X) = \infty$ otherwise. This obviously, defines a trivial generalized deviation measure.
		Therefore, we only have to show the deviation properties in case $Y$ is non-constant. For Translation Insensitivity, let $c\in\mathbb{R}$ and $X\in L^\infty$. Then \begin{align*}
		\mathcal{D}_Y(X+c)&=\inf\{m\in\mathbb{R}\colon X+c+m\in\mathcal{A}_Y\}+E[X+c]\\
		&=\inf\{m-c\in\mathbb{R}\colon X+m\in\mathcal{A}_Y\}+E[X]+c\\
		&=\inf\{m\in\mathbb{R}\colon X+m\in\mathcal{A}_Y\}+E[X]=\mathcal{D}_Y(X).
		\end{align*}
		Convexity follows because $\mathcal{A}_Y$ is Convex. 
		%Let  $\lambda\in[0,1]$, $\lambda^* = 1-\lambda$,  $X,Z\in L^\infty$  and $h = E[\lambda X + \lambda^*Z]$ we have $\lambda(\mathcal{D}_Y(X)-E[X])+\lambda^*(\mathcal{D}(Z) - E[X])
		%=\lambda\mathcal{D}_Y(X)+\lambda^*\mathcal{D}(Z) -h $ and
		%\begin{align*}
		%\lambda\mathcal{D}_Y(X)+\lambda^*\mathcal{D}(Z) -h 
		%&=\inf\{\lambda m_1+\lambda^*m_2\in\mathbb{R}\colon X+m_1\in\mathcal{A}_Y,Y+m_2\in\mathcal{A}_Y\}\\
		%&\geq \inf\{\lambda m_1+\lambda^*m_2\in\mathbb{R}\colon \lambda (X+m_1)+\lambda^*(Z+m_2)\in\mathcal{A}_Y\}\\
		%&=\inf\{\lambda m\in\mathbb{R}\colon \lambda X+\lambda^*Z+m\in\mathcal{A}_Y\}
		%\\ &=\mathcal{D}_Y(\lambda X+\lambda^*Z) -h.
		%\end{align*}
		In fact, for any $\lambda\in[0,1]$ and any $X,Z\in L^\infty$ we have \begin{align*}
		&\lambda\mathcal{D}_Y(X)+(1-\lambda)\mathcal{D}_{Y}(Z)\\
		=&\inf\{\lambda m_1+(1-\lambda)m_2\in\mathbb{R}\colon X+m_1\in\mathcal{A}_Y,Y+m_2\in\mathcal{A}_Y\}+E[\lambda X]+E[(1-\lambda)Z]\\
		\geq &\inf\{\lambda m_1+(1-\lambda)m_2\in\mathbb{R}\colon \lambda (X+m_1)+(1-\lambda)(Z+m_2)\in\mathcal{A}_Y\}+E[\lambda X+(1-\lambda)Z]\\
		=&\inf\{ m\in\mathbb{R}\colon \lambda X+(1-\lambda)Z+m\in\mathcal{A}_Y\}+E[\lambda X+(1-\lambda)Z]=\mathcal{D}_Y(\lambda X+(1-\lambda)Z).
		\end{align*}
		For Non-Negativity, we already showed that if  $c\in\mathbb{R}$ then $\mathcal{D}_Y(c) = \mathcal{D} (c)=0$.
		%\begin{align*}
		%\mathcal{D}_Y(c)&=\inf\{m\in\mathbb{R}\colon c+m\in\mathcal{A}_Y\}+E[c]\\
		%&=\inf\{m\in\mathbb{R}\colon c+m\in\mathbb{R}_+\}+c\\
		%&=-c+c=0.
		%\end{align*}
		For $X\in L^\infty$ non-constant, we have that 
		%\begin{align*}
		%	\mathcal{D}_Y(X)&=\inf\{m\in\mathbb{R}\colon X+m= \lambda(Y-E[Y]+\mathcal{D}(Y)+k),\lambda\in(0,1],k\in\mathbb{R}_+\}+E[X]\\
		%	&\geq\inf\{m\in\mathbb{R}\colon X+m\geq \lambda(Y-E[Y]+\mathcal{D}(Y)),\lambda\in(0,1]\}+E[X]\\
		%	&=\inf\limits_{\lambda\in(0,1]}\left\lbrace \esssup\{\lambda Y-X\}-E[\lambda Y]+\lambda\mathcal{D}(Y)\right\rbrace +E[X]\\
		%	&\geq\inf\limits_{\lambda\in(0,1]}\left\lbrace E[\lambda Y-X]+E[X-\lambda Y]+\lambda\mathcal{D}(Y)\right\rbrace\\
		%	& =\inf\limits_{\lambda\in(0,1]}\lambda\mathcal{D}(Y)=0.
		%\end{align*} %Since we have that $\mathcal{D}_Y(X)=\esssup\{Y-X\}-E[Y]+\mathcal{D}(Y)+E[X]$, we get that $\mathcal{D}_Y$ is finite (because we are in $L^\infty$) and lower semi continuous (by the continuity properties of $E$, $\mathcal{D}$ and $\esssup$) in both strong and weak topologies. For Fatou continuity, let $\{X_n\}$ bounded such that $\lim\limits_{n\to\infty}X_n=X$. If $X$ is constant, then $\mathcal{D}_Y(X)=0\leq\liminf\limits_{n\to\infty}\mathcal{D}_Y(X_n)$. If $X$ is not constant, then \begin{align*}\mathcal{D}(X)&=\esssup\{Y-\lim\limits_{n\to\infty}X_n\}-E[Y]+\mathcal{D}(Y)+E[\lim\limits_{n\to\infty}X_n]\\&\leq\liminf\limits_{n\to\infty}\left(\esssup\{Y-X_n\}-E[Y]+\mathcal{D}(Y)+E[X_n] \right)\\&=\liminf\limits_{n\to\infty}\mathcal{D}_Y(X_n).\end{align*}
		%The inequality is in fact strict since 
		$\mathcal{D}_Y(X)\geq\mathcal{D}(X)>0$.
		Thus, each $\mathcal{D}_Y$ is a Convex deviation measure.  Moreover, let $\mathcal{I}=\{\beta\colon L^\infty\to\mathbb{R}_+\cup\{\infty\}\colon\beta\:\text{is Convex deviation and}\:\beta\geq\mathcal{D}\}$. We have that $\mathcal{D}(X)\leq\inf_{\mathcal{I}}\beta (X)$. Since $\mathcal{D}_X\in\mathcal{I}$, we have that $\mathcal{D}(X)=\min_{\mathcal{I}}\beta(X),\:\forall\:X\in L^\infty$.
	\end{proof}
	
\begin{Rmk} 
The minimum in \cref{eq:dev_min} can not be replaced by an infimum. In order to verify it, note that it could easily lead to a situation where for a non-constant random variable $X\in L^\infty$ we may have  $\inf_{i \in \mathcal{I}} \mathcal{D}_i(X) =0$. This directly conflicts to Non-Negativity, even if all $\mathcal{D}_i$ are deviation measures.
\end{Rmk}	
	
	\begin{Rmk}\label{Rmk:set}
		The set $\mathcal{I}$ in the representation of last Theorem is not unique. Nonetheless, under a relaxation we have some uniqueness result. For any set $\mathcal{I}$ of Convex deviation measures, define its relaxation as \[\mathcal{I}^{*}=\left\lbrace \beta\colon L^\infty\to\mathbb{R}_+\cup\{\infty\}\colon\beta\:\text{is Convex deviation measure and}\:\beta\geq\min\limits_{\mathcal{I}}\beta_i\right\rbrace .\] 
		Note that $\min_{\mathcal{I}} \beta_i$ is not necessarily well-defined. However, if it is well-defined and $\mathcal{D}=\min\limits_{\mathcal{I}_1}\beta_i=\min\limits_{\mathcal{I}_2}\beta_i$, then we directly have that \[\mathcal{I}^{*}_1=\mathcal{I}_2^{*}=\{\beta\colon L^\infty\to[0,\infty]\colon\beta\:\text{is Convex deviation measure and}\:\beta\geq\mathcal{D}\}.\] 
	\end{Rmk}
	
	The \nameCref{ph,li,lrd} below provides the same result as \Cref{prp:dev} but for different classes of deviation measures. 
	
	\begin{Crl}\label{ph,li,lrd}
		Let $\mathcal{D}\colon L^\infty\to\mathbb{R}_+\cup\{\infty\}$ be a Star-Shaped deviation measure represented under $\mathcal{I}$ in the context of  \Cref{prp:dev}. Then:
		\begin{enumerate}
			\item $\mathcal{D}$ is Positive Homogeneous if and only if there exists some $\mathcal{I}$  composed by generalized deviation measures such that \cref{eq:dev_min} holds.
			%\item $\mathcal{D}$ is Law Invariant if and only if there is some $\mathcal{I}$ composed by Law Invariant Convex deviation measures such that \cref{eq:dev_min} holds.
			\item $\mathcal{D}$ is Lower Range Dominated if and only if there is some $\mathcal{I}$ composed by lower ranged deviation measures such that \cref{eq:dev_min} holds.
		\end{enumerate} 
		In any case, such families can be chosen, respectively, as the one composed by \[\mathcal{I}=\{\beta\colon L^\infty\to\mathbb{R}_+\cup\{\infty\}\colon\beta\:\text{ (PH or LRD) Convex deviation measure and}\:\beta\geq\mathcal{D}\}.\]
	\end{Crl}
	
	\begin{proof}
		For (i), if $\mathcal{I}$ is composed by generalized deviation measures, then for any $\lambda\geq 0$ and $X\in L^\infty$ we have that \[\mathcal{D}(\lambda X)=\min\limits_{i\in\mathcal{I}}\lambda\mathcal{D}_i(X)=\lambda\min\limits_{i\in\mathcal{I}}\mathcal{D}_i(X)=\lambda\mathcal{D}(X).\] Thus, $\mathcal{D}$ is Positive Homogeneous. For the converse, let $\mathcal{D}_Y$ be defined as in  \Cref{prp:dev}, but now with \[\mathcal{A}_{Y}=
		convco((\{Y-E[Y]+\mathcal{D}(Y)\}+\mathbb{R}_+)\cup\{0\})+\mathbb{R}_+,\] where convco means the Convex conic hull i.e. $convco (A) = \{k ( X +  Y )\colon X,Y \in A, k \geq 0  \} $. Clearly $\mathcal{A}_Y$ is a cone. Then, for any $X\in L^\infty$ and $\lambda\geq 0$ we get that \begin{align*}
		\mathcal{D}_Y(\lambda X)&=\inf\{m\in\mathbb{R}\colon \lambda X+m\in\mathcal{A}_Y\}+E[\lambda X]\\&=\inf\{\lambda n\in\mathbb{R}\colon \lambda (X+n)\in\mathcal{A}_Y\}+E[\lambda X]\\
		&=\lambda\left( \inf\{ n\in\mathbb{R}\colon  X+n\in\mathcal{A}_Y\}+E[ X]\right) =\lambda\mathcal{D}_Y(X).
		\end{align*}
		Then, each $\mathcal{D}_Y$ is Positive Homogeneous. The facts that $\mathcal{D}(X)=\min\{\mathcal{D}_Y(X)\colon Y \in L^\infty\}$ and $\mathcal{I}=\{\beta\colon L^\infty\to\mathbb{R}_+\cup\{\infty\}\colon\beta\:\text{generalized deviation measure and}\:\beta\geq\mathcal{D}\}$ follow as in  \Cref{prp:dev}.
		
		%Regarding (ii), if $\mathcal{I}$ is composed by Law Invariant Convex deviation measures, then for any $X\sim Y$ and $X,Y\in L^\infty$ we have that \[\mathcal{D}(X)=\min\limits_{i\in\mathcal{I}}\mathcal{D}_i(X)=\min\limits_{i\in\mathcal{I}}\mathcal{D}_i(Y)=\mathcal{D}(Y).\] 		Thus, $\mathcal{D}$ is Law Invariant. For the converse, let $\mathcal{D}_Y$ be defined as in  \Cref{prp:dev}, but now with		$\mathcal{A}_Y^\prime=\{Z\sim X\colon X\in\mathcal{A}_Y\}$, which is clearly a law-invariant set. Then, for any $X\sim W$ and $X,W\in L^\infty$ we have that 		\begin{align*}\mathcal{D}_Y( X)&=\inf\{m\in\mathbb{R}\colon  X+m\in\mathcal{A}_Y^\prime\}+E[X]\\&=\inf\{ m\in\mathbb{R}\colon W+m\in\mathcal{A}_Y^\prime\}+E[W]=\mathcal{D}_Y(W).		\end{align*}		Then, each $\mathcal{D}_Y$ is Law Invariant. The facts that $\mathcal{D}(X)=\min\{\mathcal{D}_Y(X)\colon Y \in L^\infty\}$ and $\mathcal{I}=\{\beta\colon L^\infty\to\mathbb{R}_+\cup\{\infty\}\colon\beta\:\text{Law Invariant Convex deviation measure and}\:\beta\geq\mathcal{D}\}$ follow as in  \Cref{prp:dev}.
		
		Concerning (ii), if $\mathcal{I}$ is composed by  Lower Range Dominated deviation measures, then for any  $X\in L^\infty$ we have that \[\mathcal{D}( X)=\min\limits_{i\in\mathcal{I}}\lambda\mathcal{D}_i(X)\leq E[X]-\essinf X.\] Thus, $\mathcal{D}$ is Lower Range Dominated. 
		For the converse, note that the  singleton containing only $\mathcal{D}$ itself satisfies \cref{eq:dev_min}.Hence, if $\mathcal{D}$ is  Lower Range Dominated then there is some $\mathcal{I}$ composed by lower ranged deviation measures such that \cref{eq:dev_min} holds.
		 Lastly,  		let $\mathcal{D}_Y$ be defined as in  \Cref{prp:dev}, but now with $\mathcal{A}_{Y}=
		conv((\{Y-E[Y]+\mathcal{D}(Y)\}+L^\infty_+)\cup\{0\})+L^\infty_+$.  Clearly $\mathcal{A}_Y$ is a Monotone set that contains $L^\infty_+$, which implies $X-\essinf X\in \mathcal{A}_Y$. Then, for any $X\in L^\infty$  we get that \begin{align*}
		\mathcal{D}_Y(X)=\inf\{m\in\mathbb{R}\colon X+m\in\mathcal{A}_Y\}+E[ X]\leq E[\lambda X]-\essinf X.
		\end{align*}
		Then, each $\mathcal{D}_Y$ is Lower Range Dominated. The facts that $\mathcal{D}(X)=\min\{\mathcal{D}_Y(X)\colon Y \in L^\infty\}$ and $\mathcal{I}=\{\beta\colon L^\infty\to\mathbb{R}_+\cup\{\infty\}\colon\beta\:\text{Lower Range Dominated Convex deviation measure and}\:\beta\geq\mathcal{D}\}$ follow as in  \Cref{prp:dev}.
	\end{proof}

	\begin{Rmk}\label{rmk:rev}
	The same reasoning for preservation of properties in Corollary \ref{ph,li,lrd} is not true for Law Invariance. We thank an anonymous reviewer for raising this point and providing the concrete example. If $D$ can be written as the minimum of law-invariant convex deviation measures, then $D$ is
	law-invariant and star-shaped. But not all law-invariant and star-shaped $D$ can be written
	as the minimum of law-invariant convex deviation measures. The intuitive reason is that
	law-invariant convex deviation measures respect convex order (under some continuity), so
	by taking a minimum one arrives at a law-invariant star-shaped deviation measure which is
	consistent with convex order. Note that not all law-invariant star-shaped deviation measures
	respect convex order. To see a concrete example without imposing any continuity, let $X$ be uniform on $[-2, 2]$, and
	$Y = (2 - X)1_{\{X>0\}} - (X + 2)1_{\{X<0\}}$, which is also uniform on $[-2, 2]$; $Z := X/2 + Y/2$ is
	distributed on $\{-1, 1\}$ with equal probability. Let $D$ be IQD at level $0.4$ (i.e., $0.4$-quantile
	minus $0.6$-quantile) plus SD (this is a positively homogeneous and proper deviation measure).
	We can easily compute $D(X) = D(Y ) = 4/5 + (4/3)^{1/2} < 2$ and $D(Z) = 2 + (2/3)^{1/2} > D(X).$
	If $D$ is the minimum of some law-invariant convex deviation measures $\{D_i\}_{i \in \mathcal{I}}$, then we must
	have $D(Z) \leq D(X)$, since $D_i(Z) \leq D_i(X)/2 + D_i(Y )/2 = D_i(X)$ for each $i \in \mathcal{I}$. This is a
	contradiction.
\end{Rmk}

	\cite{Moresco2021} show that the Star-Shapedness of the minimum of a family of Convex risk measures is closely related to the behavior of each measure around $0$. The same holds for deviation measures. However, a proper deviation measure is automatically Normalized as Non-Negativity implies $\mathcal{D}(0) = 0$. Hence our result that the minimum of Convex deviation measures is a Star-Shaped deviation measure, with no need for the extra assumption of Normalization found in \cite{Castagnoli2021}. It is possible to write any proper deviation measure as the minimum of a family of deviation measures that are Convex. However, to do so, we must drop Non-Negativity, just like one needs to drop Normalization to write a monetary risk measure as the minimum of Convex risk measures. See \cite{jia21}.
	The next \nameCref{not star} shows this fact. To that, let $\chi_A$ be the characteristic function of $A$, i.e. $  \chi_A (X) =0$ if $X \in A$ and $\infty$ otherwise. Note that $\chi_\mathbb{R}$ defines a trivial proper deviation measure.
	
	\begin{Prp}\label{not star}
		$\mathcal{D}\colon L^\infty\to\mathbb{R}_+\cup\{\infty\}$ is a proper deviation measure if and only if there is a family  of  Translation Insensitive deviation measures  $\{\mathcal{D}_i\}_{i\in\mathcal{I}}$ which respects Convexity and  such that  $\chi_\mathbb{R}\in \{\mathcal{D}_i\}_{i\in\mathcal{I}}$, $\mathcal{D}_i(X) =0$ holds only if $X \in \mathbb{R}$ and  such that the representation holds \begin{equation*}
		\mathcal{D}(X)=\min_{i\in\mathcal{I}}\mathcal{D}_i(X),\:\forall\:X\in L^\infty.
		\end{equation*}
		Moreover, such a family can be chosen as the one composed by the Translation Insensitive deviation measures that fulfill Convexity and dominate $\mathcal{D}$.
	\end{Prp}
	
	\begin{proof}
		To see that $\min_{i\in\mathcal{I}}\mathcal{D}_i(X)$ defines a proper deviation measure note that Non-Negativity follows from the presence of $\chi_\mathbb{R}$ in $\mathcal{I}$, i.e. for constant $X$ we have that $0 \leq \min_{i\in\mathcal{I}}\mathcal{D}_i(X) =  \chi_\mathbb{R} (X) =0$. Translation Insensitivity follow the same reasoning as \Cref{prp:dev}. The converse follows the same steps as in the proof of \Cref{prp:dev}, but defining $\mathcal{A}_Y$ as \[ \mathcal{A}_Y = \{ Y-E[Y] + \mathcal{D}(Y) \} + \mathbb{R}_+ .\]
		Translation Insensitivity, Convexity of each $\mathcal{D}_Y$ and that $\mathcal{D}_X(X) = \mathcal{D}(X)$ will follow from the same reasoning. However, we will lose the Star-Shapedness of each $\mathcal{D}_Y$, and consequently, of $\mathcal{D}$, as not all $\mathcal{A}_Y$ may be Star-Shaped. Nevertheless, more importantly, we will also lose Non-Negativity, as for any non-constant $Y$ and constant $c$ we will have that $\mathcal{D}_Y (c) = \infty$. While it follows that $\mathcal{D}_c = \chi_\mathbb{R}$ for any constant $c$.
	\end{proof}

	Based on the Lower Range Dominance, it is possible to obtain an interplay between risk and deviation measures. Theorem 2 of \cite{Rockafellar2006} ensures this for generalized and Convex deviations. Moreover, in \Cref{prp:dev_risk}, we do the same for Star-Shaped deviations. This claim is a generalization of  \cite{Rockafellar2006} Theorem 2.
	
	\begin{Def}\label{def:risk}
		A functional $\rho:L^\infty\rightarrow\mathbb{R}$ is a Star-Shaped risk measure if it has the following properties: 
		
		\begin{enumerate}
			\item Monotonicity: If $X \leq Y$, then $\rho(X) \geq \rho(Y),\:\forall\: X,Y\in L^\infty$.
			\item Translation Invariance: $\rho(X+C)=\rho(X)-C,\:\forall\: X,Y\in L^\infty,\:\forall\:C \in \mathbb{R}$.
			\item Normalization: $\rho(0)=0$.
			\item Star-Shapedness: $\rho(\lambda X)\geq\lambda \rho(X),\:\forall\: X\in L^\infty,\:\forall\:\lambda \geq 1$.
			%	\item Fatou continuity:  If $\lim\limits_{n\rightarrow\infty}X_n=X$ implies that $\rho(X) \leq \liminf\limits_{n\rightarrow\infty} \rho( X_{n})$, $\forall\:\{X_n\}_{n=1}^\infty$ bounded in $L^\infty$ norm and for any $X\in L^\infty$.
		\end{enumerate}
	\end{Def}

	\begin{Thm}\label{prp:dev_risk}
		We have that:
		\begin{enumerate}
			\item if $\rho\colon L^\infty\to\mathbb{R}$ is a Star-Shaped risk measure such that $\rho(X)>-E[X]$ for any non-constant $X$, then $\mathcal{D}(X)=\rho(X-E[X])$ is a Star-Shaped deviation measure.
			\item if $\mathcal{D}\colon L^\infty\to\mathbb{R}_+\cup\{\infty\}$ is a Star-Shaped deviation measure  such that $\mathcal{D}(X)\leq E[X]-\essinf X$ for any  $X$, then $\rho(X)=-E[X]+\mathcal{D}(X)$ is a Star-Shaped risk measure. 
		\end{enumerate}
	\end{Thm}
	\begin{proof}
		For (i), Translation Insensitivity follows by $\mathcal{D}(X+c)=\rho(X+c-E[X]-c)=\mathcal{D}(X)$. For Star-Shapedness, let $\lambda\geq 1$, then $\mathcal{D}(\lambda X)=\rho(\lambda(X-E[X]))\geq\lambda(\rho(X)-E[X])=\lambda\mathcal{D}(X)$. For Non-Negativity, if $c\in\mathbb{R}$ we have that $\mathcal{D}(c)=\rho(c-E(c))=\rho(0)=0$. If $X$ is non-constant, then $\mathcal{D}(X)=\rho(X)-(-E[X])>0$.
		
		Regarding (ii), Translation Invariance follows as $\rho(X+c)=-c-E[X]+\mathcal{D}(X+c)=\rho(X)-c$. Normalization is as $\rho(0)=-E[0]+\mathcal{D}(0)=0$. For Star-Shapedness, let $\lambda\geq 1$. Then $\rho(\lambda X)=-E[\lambda X]+\mathcal{D}(\lambda X)\leq\lambda(-E[X]+\mathcal{D}(X))=\lambda\rho(X)$. 
		For Monotonicity, let $X\geq Y$. Then for any $\lambda\in(0,1)$ there is $Z\geq X$ such that $X=\lambda Y+(1-\lambda )Z$. By \Cref{prp:dev} we have that $\mathcal{D}$ is the point-wise minimum of a family of Convex deviation measures. Thus, we get
		\begin{align*}
		\rho(X)&\leq -E(\lambda Y+(1-\lambda Z))+\min\limits_{\mathcal{I}}\left\lbrace \lambda\mathcal{D}_i(Y)+(1-\lambda)\mathcal{D}_i(Z)\right\rbrace \\&\leq -\lambda E[Y]-(1-\lambda) E[Z]+(1-\lambda)(E[Z]-\essinf Z)+\lambda\min\limits_{\mathcal{I}}\mathcal{D}_i(Y)\\
		&\leq \lambda(E[-Y]+\min\limits_{\mathcal{I}}\mathcal{D}_i(Y))-(1-\lambda)\essinf X=\lambda \rho(Y)-(1-\lambda)\essinf X.	
		\end{align*}
		Since for any $\lambda\in(0,1)$ there is $Z\geq X$ that satisfies the inequality, we then get that \[\rho(X)\leq\lim\limits_{\lambda\uparrow 1}\left( \lambda\rho(Y)-(1-\lambda)\essinf X\right) =\rho(Y).\]
	\end{proof}
	
	\begin{Rmk}
		In the conditions of the last \nameCref{prp:dev_risk}, both $\rho$ and $\mathcal{D}$ also inherit some properties such as lower semicontinuity and Law Invariance. Furthermore, one can replace the expectation in those formulations for another Star-Shaped risk measure $\mu$ as long as $\mu(X)+\mathcal{D}(X)\leq-\essinf X$ for any $X\in L^\infty$. This property is called Limitedness, and such composition is studied in \cite{Righi2018a}. Furthermore, we have  in the context of  \Cref{prp:dev} and \Cref{prp:dev_risk} an interplay of acceptance sets. Let $\mathcal{D}$ be some Star-Shaped deviation measure represented under $\mathcal{I}$ and generating the Star-Shaped risk measure $\rho$. The acceptance set for monetary risk measures is traditionally defined as $ \mathcal{A}_\rho = \{X\in L ^\infty\colon \rho(X)\leq 0\}$. Then it is easy to observe  that $\mathcal{A}_\mathcal{D}=\cup_{\mathcal{I}}\mathcal{A}_{\mathcal{D}_i}=\mathcal{A}_{\rho}$.
	\end{Rmk}

	%\begin{align*}	 \mathcal{A}_{FR} = \{ X \in L^\infty \colon UR (X) \leq \essinf X \};  &&	  \mathcal{A}_{UR} = \{ X \in L^\infty \colon \esssup X \leq  2 E[ X] \}; 	 \end{align*}	  \[  \mathcal{A}_{LR} = \{ X \in L^\infty \colon ML (X) \leq  0\} ,\] 		note that $\mathcal{A}_{LR} = \mathcal{A}_{ML} = \{ X \in L^\infty \colon ML(X) \leq 0\} $ where the latter is the traditional acceptance set of the monetary risk measure Maximum Loss (ML). Its acceptance set coincides with the acceptance set of the Loss VaR i.e. \[ \mathcal{A}_{LVaRD} = \left\lbrace X \in L^\infty \colon \sup\limits_{u \geq 0}\{VaR^{\alpha(u)}(X) -u\} \leq 0 \right\rbrace. \]We have that \[\mathcal{A}_{\mathcal{D}_f} = \{ X \in L^\infty \colon f(X) \leq 0 \}.\]  have that  \begin{align*}	 \mathcal{A}_{LR_f} = \{ X \in L^\infty \colon  LR_f(X)\ \leq   E [X] \}; &&	  \mathcal{A}_{UR_f} = \{ X \in L^\infty \colon UR_f (X) \leq  E[ X] \}.	 \end{align*}Its acceptance set is \[\mathcal{A}_{\mathcal{D}} = \{ X \in L^\infty \colon  \min (UD_f (X), LD_f (X)) \leq E[X] \} =  \mathcal{A}_{LR_f} \cup  \mathcal{A}_{UR_f}.\] Its acceptance set, in our framework, is		\[ \mathcal{A}_{MD_\mathcal{A}} = \left\lbrace X \in L^\infty \colon \dfrac{X}{E[X] } \in \mathcal{A}, \; E[X] > 0  \right\rbrace \bigcup \left\lbrace 0 \right\rbrace . \]
	
We have as a direct corollary that under Lower Range Dominance it is possible to extend the last claim in \Cref{accep set}. 
	
\begin{Crl}
Let $\mathcal{D}\colon L^\infty\to\mathbb{R}_+\cup\{\infty\}$ be a Lower Range Dominated Star-Shaped deviation measure and consider the notation in \Cref{accep set}. Then $cl(\mathcal{A})=\mathcal{A}_{\mathcal{D}_\mathcal{A}}$.
\end{Crl}
\begin{proof}
	Under Lower Range Dominance, we get by \Cref{prp:dev_risk} that $\rho(X)=-E[X]+\mathcal{D}(X)$ defines a Star-Shaped risk measure for any Star-Shaped deviation $\mathcal{D}$. In particular, $\mathcal{D}_\mathcal{A}(X)=\rho_\mathcal{A}(X)+E[X]$. Then, Proposition 4.3 in \cite{Follmer2016} assures that $cl(\mathcal{A})=\{X\in L^\infty\colon\rho_\mathcal{A}(X)\leq 0 \}$. Hence, $cl(\mathcal{A})=\mathcal{A}_{\mathcal{D}_\mathcal{A}}$.
\end{proof}

	In the context of last \nameCref{prp:dev} there is also an interplay between increasing Convex order for $\rho$ and Convex order for $\mathcal{D}$, that allows us to prove results that mimic Theorems 11 and 12 in \cite{Castagnoli2021}.

	\begin{Crl}\label{Thm:LI2}
		Let the probability space be atomless and $\mathcal{D}\colon L^\infty\to\mathbb{R}_+\cup\{\infty\}$ a Lower Range Dominated Star-Shaped deviation measure. Then:
		\begin{enumerate}
			\item $\mathcal{D}$ is
			Law Invariant if and only if there is a Star-Shaped set $G$ of non-increasing functions  $g\colon(0,1)\to\mathbb{R}$ with $g(1^-)\geq 0$ such that	\begin{equation}\label{eq:dual18}
			\mathcal{D}(X)=\inf_{g\in G}\sup\limits_{\alpha\in(0,1)}\left\lbrace VaR^\alpha(X-E[X])-g(\alpha) \right\rbrace,\:\forall\:X\in L^\infty. 
			\end{equation}
			\item $\mathcal{D}$ is Convex order consistent if and only if  there is a Star-Shaped set $G$ of non-increasing functions  $g\colon(0,1)\to\mathbb{R}$ with $g(1^-)\geq 0$ such that
			\begin{equation}\label{eq:dual17}
			\mathcal{D}(X)=\inf_{g\in G}\sup\limits_{\alpha\in(0,1)}\left\lbrace ES^\alpha(X-E[X])-g(\alpha) \right\rbrace ,\:\forall\:X\in L^\infty. 
			\end{equation}
		\end{enumerate} 
	\end{Crl}
	
	\begin{proof}
		\begin{enumerate}
			\item It is straightforward to verify that \eqref{eq:dual18} defines a law-invariant Star-Shaped deviation measure. Lower Range Dominance follows from the Monotone behavior in $\alpha$ for VaR and the functions in $G$ as \[\mathcal{D}(X)\leq VaR^0(X-E[X])-\sup_{g\in G}g(1-) \leq E[X]-\essinf X,\:\forall\:X\in L^\infty.\] For the converse, we have by \Cref{prp:dev_risk} that $\rho(X)=-E[X]+\mathcal{D}(X),\:X\in L^\infty$ is a Law Invariant Star-Shaped risk measure, which implies it is Monotone regarding to increasing order. Thus, by Theorem 12 of \cite{Castagnoli2021} there is a Star-Shaped set $G$ of non-increasing functions  $g\colon(0,1)\to\mathbb{R}$ with $g(1^-)\geq 0$ such that for any $X\in L^\infty$ that \[\rho(X)=\inf_{g\in G}\sup\limits_{\alpha\in(0,1)}\left\lbrace VaR^\alpha(X)-g(\alpha) \right\rbrace.\] By adding $E[X]$ in both sides of the last equation we get the claim. 
			\item That \eqref{eq:dual18} defines a Lower Range Dominated Law Invariant Star-Shaped deviation measure it is similar to item (i) since $\alpha\to ES^\alpha(X)$ is Monotone for any $X\in L^\infty$. Let $X,Y\in L^\infty$ such that $X\succeq Y$ in Convex order. Then we have that $E[X]=E[Y]$ and $ES^\alpha(X)\leq ES^\alpha(Y)$ for all $\alpha \in (0,1)$. This directly implies  $\mathcal{D}(X)\leq \mathcal{D}(Y)$. For the converse, we have by  \Cref{prp:dev_risk} that $\rho(X)=-E[X]+\mathcal{D}(X),\:X\in L^\infty$ is a Law Invariant Star-Shaped risk measure. We claim that $\rho$ is consistent with respect to increasing Convex order. Let $X\succeq Y$ in such order. Then $X-E[X]\succeq Y-E[Y]$ in Convex order. Thus, $E[X]\geq E[Y]$ and $\mathcal{D}(X)=\mathcal{D}(X-E[X])\leq \mathcal{D}(Y-E[Y])=\mathcal{D}(Y)$. This directly implies $\rho(X)\leq \rho(Y)$. Now, by Theorem 12 of \cite{Castagnoli2021} there is a Star-Shaped set $G$ of non-increasing functions  $g\colon(0,1)\to\mathbb{R}$ with $g(1^-)\geq 0$ such that for any $X\in L^\infty$,
			  \[\rho(X)=\inf_{g\in G}\sup\limits_{\alpha\in(0,1)}\left\lbrace ES^\alpha(X)-g(\alpha) \right\rbrace.\] By adding $E[X]$ in both sides of the last equation we get the claim.
		\end{enumerate}
	\end{proof}

It is shown in \cite{Mao2020} that minimimum of Law Invariant Convex risk measures are precisely the Second Stochastic Dominance (SSD) consistent 
risk measures. We now show that a similar result holds for  Lower Range Dominated Star-Shaped deviation measures where SSD should be replaced by Convex order.

\begin{Crl}\label{Thm:LI3}
Let the probability space be atomless. $\mathcal{D}\colon L^\infty\to\mathbb{R}_+\cup\{\infty\}$ is a Convex order consistent Lower Range Dominated Star-Shaped deviation measure there exists some $\mathcal{I}$  composed by Law Invariant convex deviation measures such that \cref{eq:dev_min} holds.
\end{Crl}

\begin{proof}
On atomless probability spaces any Law Invariant Convex Deviation is consistent to Convex order, see \cite{Grechuk2009} for instance. It is straightforward to verify that the same is true for the supremum, the infimum, or any convex combination of Law Invariant Convex Deviations. Moreover, the infimimum in representation \eqref{eq:dual17} of item (ii) of Corollary \ref{Thm:LI2} is attained by taking the star-shaped set $G=\{g_Y(\alpha)= ES^\alpha(Y)\colon E[-Y]+D(Y)\leq 0\}$, see Theorem 3.1 in \cite{Mao2020} and \cite {Castagnoli2021} for details. Hence, under such reasoning, the claim is directly obtained. 
\end{proof}

\section*{Acknowledgments}

We would like to thank an anonymous reviewer for helpful comments. In particular, the content in \Cref{rmk:rev} was raised by this anonymous reviewer. We are grateful for the financial support of CNPq (Brazilian Research Council) project number
302614/2021-4.
	\bibliography{ref}
	\bibliographystyle{abbrv}
\end{document}